\documentclass[journal,twoside,web]{ieeecolor}
\usepackage{generic}
\usepackage{cite}
\usepackage{color}

\definecolor{purple}{rgb}{0.75,0.194,0.34}
\usepackage{amsmath,amssymb,amsfonts}
\usepackage{algorithmic}
\usepackage{graphicx}
\usepackage{algorithm,algorithmic}

\makeatletter
\let\NAT@parse\undefined
\makeatother

\usepackage[colorlinks=true,
    linkcolor=blue,
    citecolor=blue]{hyperref}

\newtheorem{theo}{Theorem}
\newtheorem{assu}{Assumption}
\newtheorem{lemma}{Lemma}

\newtheorem{remark}{Remark}

\usepackage{textcomp}
\def\BibTeX{{\rm B\kern-.05em{\sc i\kern-.025em b}\kern-.08em
    T\kern-.1667em\lower.7ex\hbox{E}\kern-.125emX}}
\begin{document}
\title{Adaptive Output Feedback MPC with Guaranteed Stability and Robustness}
\author{Anchita Dey*, \IEEEmembership{Graduate Student Member, IEEE}, and Shubhendu Bhasin, \IEEEmembership{Member, IEEE}
\thanks{This paragraph of the first footnote will contain the date on 
which you submitted your paper for review. Anchita Dey is supported by the Prime Minister's Research Fellow Scheme under the Ministry of Education, Government of India.}
\thanks{Anchita Dey and Shubhendu Bhasin are with the department of Electrical Engineering, Indian Institute of Technology Delhi, Hauz Khas, New Delhi, Delhi 110016, India {\tt\small anchita.dey@ee.iitd.ac.in, sbhasin@ee.iitd.ac.in.}}
\thanks{*Corresponding author}
}

\maketitle
\begin{abstract}
This work proposes an adaptive output feedback model predictive control (MPC) framework for uncertain systems subject to external disturbances. In the absence of exact knowledge about the plant parameters and complete state measurements, the MPC optimization problem is reformulated in terms of their estimates derived from a suitably designed robust adaptive observer. The MPC routine returns a homothetic tube for the state estimate trajectory. Sets that characterize the state estimation errors are then added to the homothetic tube sections, resulting in a larger tube containing the true state trajectory. The two-tier tube architecture provides robustness to uncertainties due to imperfect parameter knowledge, external disturbances, and incomplete state information. Additionally, recursive feasibility and robust exponential stability are guaranteed and validated using a numerical example.
\end{abstract}


\section{Introduction}
\IEEEPARstart{M}{odel}  predictive control (MPC) is widely used to make optimal control decisions while adhering to hard constraints on system states and inputs. The decisions are made by optimizing a cost function dependent on predicted future states and inputs over some finite time interval.
Typically, the state predictions are obtained by propagating the current state measurement through the system dynamics {\color{black}{\cite[Ch.~2]{kouvaritakis2016model}, \cite{zhunew}}}, thus necessitating exact model knowledge and state measurements. However, obtaining a good model and reliable state information can be expensive and, in some cases, impossible. Further, predictions made using an inexact model and erroneous state measurement lead to errors that can potentially render the classical MPC solution infeasible.

Extensive literature exists on MPC with incomplete knowledge of system dynamics and/or state measurements. Robust tube-based algorithms have been developed in \cite{langson2004robust,rakovic2012homothetic} for systems with known parameters and unknown additive disturbances belonging to a known bounded set. Adaptive MPC methods, in contrast, are proposed for unknown systems with the emphasis on identifying the system \cite{tanaskovic2014adaptive} or learning the parameters \cite{dhar2021indirect,zhu2019constrained,heirung2017dual} and/or the set containing the parameters \cite{lorenzen2019robust,lu2023robust,sasfi2023robust} to improve the transient performance while robustly accounting for any ensuing uncertainty in the design of the constrained optimal control problem (COCP). Few results 
\cite{mayne2006robust,mayne2009robust,kogel2017robust,brunke2021rlo,brunner2018enhancing,de2022robust}
exist that deal with the MPC problem in the absence of full state measurements. A Luenberger observer is used in \cite{mayne2006robust,mayne2009robust,kogel2017robust,brunke2021rlo} to obtain state estimates of a system with additive disturbances and known parameters. The COCP reformulated by leveraging these state estimates returns robust tubes containing the true state trajectory. The method in \cite{brunke2021rlo} and \cite{rosolia2017learning} is based on an iterative learning approach to learn the safe region for a given repetitive task. The theory developed in \cite{brunner2018enhancing} is based on set-valued moving horizon estimation wherein an undisturbed nominal model is used to carry out predictions of nominal states and sets containing the worst-case deviations from the nominal predictions. An interval observer-based approach is proposed in \cite{de2022robust} for linear time-invariant (LTI) systems with known parameters but bounded disturbance and noise. 

The major challenges in relaxing the restrictive assumptions on model knowledge and state measurements are ensuring constraint satisfaction, recursive feasibility, and robust stability. Unlike \cite{langson2004robust, rakovic2012homothetic,lorenzen2019robust,dhar2021indirect,mayne2006robust,mayne2009robust,kogel2017robust,brunke2021rlo, rosolia2017learning,zhu2019constrained,lu2023robust,heirung2017dual,sasfi2023robust,brunner2018enhancing,de2022robust}, some recent works \cite{tanaskovic2014adaptive,ping2020observer,hu2019output,berberich2020data,coulson2019data,anch} attempt to relax both the assumptions. While \cite{tanaskovic2014adaptive} is only applicable for stable systems, \cite{ping2020observer} considers linear parameter varying systems where the current value of the parameters is exactly known. In \cite{hu2019output}, the parametric uncertainty assumed to be norm-bounded by one is handled robustly. 
In \cite{berberich2020data} and \cite{coulson2019data}, the authors represent the system behavior implicitly using a Hankel matrix that is constructed or learnt using sufficiently rich input data prior to running the robust MPC COCP. 

{\color{black}{Unlike the output feedback MPC approach discussed in \cite{hu2019output}
which robustly handles system uncertainties without attempting to learn the parameters,}} or the methods presented in \cite{berberich2020data,coulson2019data} that involve a pre-MPC routine to capture the input-output system behavior, the solution in \cite{anch} is based on online adaptation of system parameters and state via an adaptive observer design, and a suitably formulated COCP with redefined constraints. This approach allows constraint imposition on all states (not just the outputs) and inputs. However, this framework is limited to single-input single-output (SISO) LTI systems and does not provide any stability and robustness guarantees, which restricts its practical applicability.

In this paper, we extend the work of \cite{anch} to parametrically uncertain discrete-time multi-input multi-output (MIMO) LTI systems subject to bounded additive disturbances. The unknown parameters are assumed to belong to a known polytope. Motivated by \cite{suzuki1980discrete,dey2023adaptive}, a robust adaptive observer using a projection modification-based learning law \cite[Sec.~4.11.3]{ioannou2006adaptive} is designed for the MIMO case. In the absence of exact information about the state and parameters, the COCP is reformulated in terms of the estimates obtained from the observer. Instead of using a single invariant set for state estimation errors, as in \cite{anch}, we compute suitable time-dependent sets for tightening the state constraint; this helps increase the initial feasible region for the COCP. Further, it is non-trivial to predict the observer states without complete model information. This issue is overcome by propagating the robust adaptive observer dynamics, which is divided into known and unknown terms; the latter leads to errors in predictions of the observer state 
which are then accounted for by using a homothetic tube \cite{rakovic2012homothetic}. Solving the COCP in a receding horizon fashion \cite{kouvaritakis2016model} results in the homothetic tube to which the sets for state estimation errors are added to obtain a two-tube architecture. The proposed approach is named Robust Adaptive Output Feedback MPC (RAOFMPC).

The proposed framework using the two-tube structure can handle additive disturbances without using state measurement and exact model parameter knowledge and enables constraint imposition on all states and control inputs. Using separate sets for the state estimation errors leads to reduced conservatism {{as compared to \cite{anch}.}} Further, the proposed approach is shown to guarantee recursive feasibility of the COCP, and robust exponential stability \cite[Def.~2]{mayne2006robust} of the MIMO LTI system.

\subsection{Notations} $||\cdot||$ and $||\cdot||_\infty$ denote $2$-norm and $\infty$-norm, respectively. With two sets $\mathcal{C},\;\mathcal{D}\subseteq\mathbb{R}^{n}$, the Minkowski sum $\mathcal{C}\oplus\mathcal{D}\triangleq \{c+d\;|\;c\in\mathcal{C},\;d\in\mathcal{D}\}$ and $c\oplus\mathcal{D}\triangleq \{c+d\;|\; d\in\mathcal{D}\},$ where $c\in\mathbb{R}^{n}$, the Pontryagin difference $\mathcal{C}\ominus\mathcal{D}\triangleq\{c\;|\;c+d\in\mathcal{C}\;\forall d\in\mathcal{D}\}$, and $\text{\textbf{co}} (\mathcal{C})$ represents the convex hull of all elements in $\mathcal{C}$. The notation $\{\alpha(i)\}_{i=c:d}$ where $\alpha(i)$ is a function of $i$, and $c$, $d>c$ are integers, represents the set $\{\alpha(c), \alpha(c+1),...,\alpha(d)\}$. Any matrix $W\succ 0$ $(\succeq 0, \; \prec 0,\;\preceq 0)$ denotes $W$ is symmetric positive-definite (positive semi-definite, negative-definite, negative semi-definite), and $||c||^2_W\triangleq c^\intercal Wc$, where $c$ is a vector. $\mathbb{I}_c^d$ is the set of all integers from $c$ to $d$. The identity matrix and zero matrix are denoted by $I_c\in\mathbb{R}^{c\times c}$ and $0_{c\times d}\in\mathbb{R}^{c\times d}$, respectively. A bounded signal $c$ is represented as $c\in \mathcal{L}_\infty$, and $c_{i|t}$ denotes the value of $c$ at time $t+i$ predicted at time $t$. Kronecker product is denoted by $\otimes$, and eigenvalues by $\lambda(\cdot)$. The superscript $^*$ on any term denotes its optimal value, obtained by solving the COCP. For any $c\in\mathbb R^n$, $c^{(i)}$ denotes its $i^{\text{th}}$ row element. For a matrix $W=[ W_1^\intercal \;\; W_2^\intercal \;...\;W_c^\intercal]^\intercal\in\mathbb{R}^{c\times d}$, where $W_i^\intercal\in\mathbb{R}^d$ $\forall i\in\mathbb{I}_1^c$, the operator $vec(W)\triangleq 
   [ W_1 \;\; W_2 \;...\;W_c]^\intercal\in\mathbb{R}^{cd}$ and $vec^{-1}(\cdot,{\color{black}{\cdot})}$ denotes the inverse operation, i.e., $vec^{-1}\left( [ W_1 \;\; W_2 \;...\;W_c]^\intercal{\color{black}{,d }} \right)\triangleq W$. 

\section{Problem Formulation}
Consider the constrained discrete-time LTI system 
\begin{align}
    &x_{t+1}=Ax_t+Bu_t+d_t,\;\;\;y_t=Cx_t, \label{osys1}\\
    &x_t\in\mathbb{X},\;\;u_t\in\mathbb{U}\;\;\;\;\forall t\in\mathbb{I}_0^\infty ,\label{hc1}
    \end{align}
with the following observable canonical realization \cite[Ch.~4]{chen1984linear}
\begin{gather}\label{A_matrix}\begin{aligned}
&A=\left[
    \begin{array}{c|c}
   {\mathcal{A}}&\begin{array}{cc}
         I_{n-q}\\0_{q\times (n-q) }
    \end{array}
    \end{array}
    \right]\in\mathbb{R}^{n\times n},\;\;B\in\mathbb{R}^{n\times m} \text{ and}\\
    & C=\begin{bmatrix}    I_q & 0_q & ... & 0_q   \end{bmatrix}\in\mathbb{R}^{q\times n},
\end{aligned}\end{gather}
where 
$x_t\in\mathbb{R}^n$, $u_t\in\mathbb{R}^m$, $y_t\in\mathbb{R}^q$ and $d_t\in\mathbb{R}^n$ denote the state, input, output and external disturbance, respectively, at time $t$, and 
$\mathcal{A}\in\mathbb{R}^{n\times q}$ is an unknown matrix used in the observable canonical representation of $A$. The constraint sets on $x_t$ and $u_t$ are represented as $\mathbb{X}$ and $\mathbb{U}$, respectively. The unknown disturbance $d_t$ belongs to a set $\mathbb{D}$. The sets $\mathbb{X}$, $\mathbb{U}$ and $\mathbb{D}$ are assumed to be known convex polytopes containing their respective origins. The unknown terms in $A$ and $B$ are written as parameter vectors 
$a\triangleq vec(\mathcal{A})
\in\mathbb{R}^{qn}$ and $b\triangleq vec(B)\in\mathbb{R}^{mn}$, respectively.
The objective is to design a suitable control input to drive the state of \eqref{osys1} to the origin while satisfying hard constraints \eqref{hc1}, in the presence of the disturbance $d_t$. In the ideal case, when there are no model uncertainties, the following classical MPC framework achieves the objective:
\begin{align*}
\min_{\mu_t}\;J(x_t,\mu_t)\triangleq  \textstyle\sum_{i=0}^{N-1}\left(||x_{i|t}||^2_{{Q}}+||u_{i|t}||_R^2\right) +||x_{N|t}||^2_{\bar{P}}\\
\text{subject to }\; {{x_{0|t}=x_t,\;x_{i+1|t}=Ax_{i|t}+Bu_{i|t}\;\; \forall i\in\mathbb{I}_{0}^{N-1}}} ,\\
  x_{i|t}\in\mathbb{X}\text{, } u_{i|t}\in\mathbb{U}\;\;\forall i\in\mathbb{I}_{0}^{N-1},\text{ and }
 x_{N|t}\in\mathbb{X}_\text{TS}\subseteq\mathbb{X},
\end{align*}
where $\mu_t\triangleq\{u_{0|t}, u_{1|t},...,u_{N-1|t}\}$ is the input sequence, $N$ is the prediction horizon length, $Q$, $\bar{P}\in\mathbb{R}^{n\times n}$, $R\in\mathbb{R}^{m\times m}$ with ${Q}$, $\bar{P},R\succ 0$, and $\mathbb{X}_{\text{TS}}$ is the terminal set. In the absence of knowledge of the system parameters, unavailability of full state measurement, and the presence of external disturbance, it is challenging to formulate a suitable MPC optimization problem and prove stability/convergence and recursive feasibility. 

In this paper, a robust adaptive observer is developed for joint estimation of the system state and parameters while being robust to external disturbances. The proposed COCP is constructed in terms of the state and parameter estimates returned by the observer, and a robust tube-based approach, motivated by the development in \cite{anch}, is used to guarantee stability and recursive feasibility. 

For our convenience in developing the theory, matrices structurally similar to $A$ and $B$ are constructed from parameter vectors, wherever necessary; this operation is performed using the following two functions {\color{black}{$\mathcal{M}_A(\cdot,\cdot)$}} and {\color{black}{$\mathcal{M}_B(\cdot,\cdot)$}}:
{{\begin{align*}
    \mathcal{M}_A:\; \mathbb{R}^{qn} \times {\color{black}{\mathbb{I}_1^\infty}} &\rightarrow \mathbb{R}^{n\times n}\\
      {\color{black}{(\hat{a},n)}} &\mapsto 
\left[
    \begin{array}{c|c}
   {\color{black}{vec^{-1}(\hat{a},n)}}&\begin{array}{cc}
         I_{n-q}\\0_{q\times (n-q) }
    \end{array}
    \end{array}
    \right],\\
    \mathcal{M}_B: \mathbb{R}^{mn}\times {\color{black}{\mathbb{I}_1^\infty }}&\rightarrow \mathbb{R}^{n\times m}\\
    {\color{black}{(\hat{b},n)}}& \mapsto {\color{black}{vec^{-1}(\hat{b},m)}}.
\end{align*}}}

The following standard assumption is considered on the parametric uncertainties \cite{lorenzen2019robust,dhar2021indirect,anch,darup2016computation,lu2023robust}.
\begin{assu}\label{ABa}
The unknown parameter $\psi\triangleq\begin{bmatrix}a^{\intercal} &b^{\intercal} \end{bmatrix}^{\intercal} $ belongs to a set $\Psi\triangleq \text{\textbf{co}} ( \{\psi^{[i]}\;{|}\; i\in\mathbb{I}_1^L \})$$\subseteq \mathbb{R}^{qn+mn}$,
whose vertices  $\psi^{[1]},\psi^{[2]},...,\psi^{[L]}$ are known and $L$ is some finite positive integer, with each element in $\mathbb{M}_\Psi\triangleq {\{} (\mathcal{M}_A  (\hat{a},n  ),\mathcal{M}_B  ( \hat{b},n  ) ) \; {|} \; \begin{bmatrix}\hat{a}^{\intercal} &{\hat{b}}^{\intercal} \end{bmatrix}^{\intercal} \in\Psi,\;\hat{a}\in\mathbb{R}^{qn},\;\hat{b}\in\mathbb{R}^{mn} {\}}$ being {\color{black}{controllable. In addition, corresponding to the user-defined matrices $Q$, $R\succ 0$, $\exists$ a pair $(P,K)$ where $P\in\mathbb{R}^{n\times n}$ and $K\in\mathbb{R}^{m\times n}$ such that $\forall (\widehat{A},\widehat{B})\in\mathbb{M}_\Psi$}}
\begin{align}
 {\color{black}{P-\left( \widehat{A}+\widehat{B}K \right)^{\intercal} P \left( \widehat{A}+\widehat{B}K \right)-Q-K^{\intercal}RK\succeq 0.}} \label{Asmeq}
\end{align}
\end{assu}

Assumption \ref{ABa} is required for constraint tightening, terminal set construction and stability analysis. The pair $(P, K)$ can be computed using linear matrix inequalities \cite{de1999new}.
\section{Robust Adaptive Observer for MIMO LTI Systems}

Let $F\in\mathbb{R}^{n\times n}$ be a Schur stable matrix that is structurally similar to $A$, i.e., $F=\mathcal{M}_A(f,n)$, where
$f\in\mathbb{R}^{qn}$ is user-defined. The plant dynamics in \eqref{osys1} can be rewritten as
\begin{align}\label{osys1r}
    x_{t+1}=Fx_t+Y_t(a-f)+U_tb+d_t,
\end{align}
where $Y_t\triangleq I_n\otimes y_t^{\intercal}\in\mathbb{R}^{n\times qn}$ and $U_t\triangleq I_n\otimes u_t^{\intercal}\in\mathbb{R}^{n\times mn}$. Define two new vectors $p\triangleq \begin{bmatrix}
    (a-f)^{\intercal} & b^{\intercal}
\end{bmatrix}^{\intercal}\in\mathbb{R}^{qn+mn}$, which is unknown, and its estimate $\hat{p}_t\triangleq\begin{bmatrix}(\hat{a}_t-f)^{\intercal}   & \hat{b}_t^{\intercal}   \end{bmatrix}^{\intercal}  \in\mathbb{R}^{qn+mn}$, where $\hat{a}_t$ and $\hat{b}_t$ are the estimates of $a$ and $b$, respectively. Also, define a filter variable $M_t\in\mathbb{R}^{n\times(qn+mn)}$ $\forall t\in\mathbb{I}_0^{\infty}$ satisfying
\begin{align}\label{FilterM}
    M_{t+1}=FM_t+\begin{bmatrix}
        Y_t & U_t
    \end{bmatrix}\;;\;M_0=0_{n\times(qn+mn)}.
\end{align}
Using $M_t$, the state and output of \eqref{osys1r} are expressed as
\begin{align}
    &{x}_t=M_tp+F^t{x}_0+\textstyle\sum_{k=0}^{t-1}F^{t-1-k}d_k, \label{xtrue}\\
    &{y}_t=Cx_t=\phi^{\intercal}_tp+CF^t{x}_0+\eta_t, \label{ytrue}
\end{align}
respectively, where the regressor $\phi_t\triangleq M^{\intercal}_tC^{\intercal}  \in\mathbb{R}^{(qn+mn)\times q}$ and $\eta_t\triangleq$ $ C\sum_{k=0}^{t-1}F^{t-1-k}d_k\in\mathbb{R}^q$. Corresponding to \eqref{xtrue} and \eqref{ytrue}, the adaptive observer is designed as
\begin{align}\label{xhat}
    \hat{x}_t=M_t\hat{p}_t+F^t\hat{x}_0 \text{ and }\hat{y}_t=\phi^{\intercal}  _t\hat{p}_t+CF^t\hat{x}_0,
\end{align}
where $\hat{x}_t$ denotes the observer state and $\hat{y}_t$ denotes the observer output. {\color{black}{The filter variable $M_t$ helps obtain the linear regression forms \eqref{ytrue}, \eqref{xhat} that are amenable for designing a suitable parameter estimation law for $\hat{p}_t$, which in turn dictates the trajectories of $\hat{x}_t$ and $\hat{y}_t$ following \eqref{xhat}.}} 
The adaptive observer dynamics obtained using \eqref{FilterM} and \eqref{xhat} is given by
\begin{align}
    \hat{x}_{t+1}\;&=F\hat{x}_t+\begin{bmatrix}
        Y_t &U_t
    \end{bmatrix}\hat{p}_{t+1}\nonumber\\&+\textstyle\sum_{k=0}^{t-1}F^{t-k}\begin{bmatrix}
        Y_k &U_k
    \end{bmatrix}(\hat{p}_{t+1}-\hat{p}_t)\;\;\;\forall t\in\mathbb{I}_0^\infty. \label{adob0}
\end{align}  

\subsection{Adaptation law for the observer}
Let $\tilde{x}_t\triangleq x_t-\hat{x}_t\in\mathbb{R}^n$ be the state estimation error, $\tilde{y}_t\triangleq
 y_t-\hat{y}_t\in\mathbb{R}^q$ the output estimation error and $\tilde{p}_t\triangleq p-\hat{p}_t\in\mathbb{R}^{qn+mn}$ the parameter estimation error. From \eqref{xtrue}-\eqref{xhat}, the state and output estimation errors are, respectively, given by
\begin{align}
   & \tilde{x}_t=M_t\tilde{p}_t+F^t\tilde{x}_0+{\textstyle\sum_{k=0}^{t-1}F^{t-1-k}d_k},\label{xtilde}\\
& \tilde{y}_t=\phi_t^{\intercal}\tilde{p}_t + CF^t\tilde{x}_0+\eta_t.\label{yhattilde}
\end{align}

In the absence of disturbance $d_t$, $\tilde{x}_t$ would eventually converge to zero if $\tilde{p}_t$ converges to zero, since the second term in \eqref{xtilde} is exponentially decaying. In the presence of $d_t$, it can be shown that $\tilde{x}_t$ and $\tilde{p}_t$ are bounded for all time, using a suitably designed adaptation law for $\hat{p}_t$. {\color{black}{We use a gradient descent law designed to minimize $\tilde{y}_t$ to obtain the estimate $\hat{p}_t$. However, due to the disturbances, the obtained $\hat{p}_t$ may lie outside the set $\Pi\triangleq \{ [(\hat{a}-f)^{\intercal}\;\;\hat{b}^{\intercal}]^\intercal \;|\; [ \hat{a}^{\intercal}\;\;\hat{b}^{\intercal}]^{\intercal}\in\Psi,$ $\hat{a}\in\mathbb{R}^{qn},\hat{b}\in\mathbb{R}^{mn} \}$
implying $\hat{\psi}_t \notin \Psi,$ where $\hat{\psi}_t$ is the estimate of the true parameter $\psi$, whereas by Assumption \ref{ABa}, we know that $\psi\in\Psi$ \footnote{{\color{black}{By definition, the mapping between elements of $\Psi$ and $\Pi$ is a bijection.}}}. We, therefore, modify the update law using a projection operation to guarantee robustness to external disturbances.}} The projection-modified normalized gradient descent parameter estimation law for the adaptive observer is given by
\begin{align}
 & \bar{p}_{t}\triangleq \hat{p}_{t-1}+{\kappa \phi_t{(I_q+\phi_t^{\intercal} \phi_t)}^{-1} (y_t-CF^t\hat{x}_0-\phi_t^{\intercal} \hat{p}_{t-1})},\nonumber\\
 &   \hat{p}_{t}=\begin{cases}
 \bar{p}_{t},\;&\text{if $\bar{p}_{t}\in\Pi$} \\
    \underset{\xi\in \Pi}{\mathrm{arg\;min}} ||\bar{p}_t-\xi||,\;&\text{otherwise}
    \end{cases}\;\;\;\forall t\in\mathbb{I}_1^\infty, \label{gd1}
\end{align}
where $0<\kappa<2$ \cite[Sec.~4.11.3]{ioannou2006adaptive}, and $\hat{p}_0$ is chosen from $\Pi$. The value of $\hat{x}_t$ is simultaneously obtained using $\hat{p}_t$ in \eqref{xhat}. The estimates of $A$ and $B$ used for prediction in the subsequently formulated COCP are given by $\widehat{A}_t\triangleq \mathcal{M}_A (\hat{a}_t,n ) \text{ and }\widehat{B}_t\triangleq \mathcal{M}_B (\hat{b}_t,n ),$
where $\hat{a}_t$ and $\hat{b}_t$ are obtained from $\hat{p}_t$.

\subsection{Convergence and Boundedness}\label{secOd}
  Since $d_t\in\mathbb{D}$ and $F$ is Schur stable, we can define a bounded set $\Omega^d_\infty \ni\lim_{t\rightarrow\infty}\sum_{i=0}^{t-1}F^{t-1-i}d_i$, where $\Omega^d_\infty \subset\mathbb{R}^n$. The set $\Omega^d_\infty$ is the minimal robust positively invariant (RPI) set satisfying $F\Omega^d_\infty\oplus \mathbb{D}\subseteq \Omega^d_\infty$. For tractable computation, instead of $\Omega^d_\infty$, its outer RPI approximation $\Omega^d$ \cite{rakovic2005invariant} is used to obtain the set $\Omega^\eta\triangleq C\Omega^d\subset\mathbb{R}^q$ that contains $\eta_t$ $\forall t\in\mathbb{I}_0^\infty$. A finite uniform upper bound for $||\eta_t||$ is given by $\Delta_1\triangleq \max_{\eta\in\Omega^\eta} ||\eta||.$
For computing an upper bound of the $||CF^t\tilde{x}_0||$ term in \eqref{yhattilde}, the following standard assumption concerning the uncertainty in the value of $x_0$ is considered \cite{mayne2006robust,mayne2009robust,kogel2017robust,hu2019output,de2022robust,brunner2018enhancing}.
\begin{assu}\label{Awo1}
The initial state estimation error $\tilde{x}_0$ belongs to a known convex polytope $\widetilde{\mathbb{X}}_0\subset \mathbb{X}$, $\widetilde{\mathbb{X}}_0$ contains the origin, and, {\color{black}{$F\widetilde{\mathbb{X}}_0\oplus \left(\Omega^{yu}_1 \oplus\mathbb{D} \right)\subseteq \widetilde{\mathbb{X}}_0$ (RPI condition)}}, where
\begin{align} 
{\color{black}{\Omega^{yu}_1 \triangleq   \{\begin{bmatrix}
        Y&U
\end{bmatrix} \tilde{p}\; {|} \; Y\in\mathbf{Y},\;U\in\mathbf{U},\;\tilde{p}\in\widetilde{\Psi} \}}}, \label{Omega_yu_1}  
\end{align}
$\mathbf{Y}\triangleq \left\{ I_n\otimes y^{\intercal}\;|\; y\in C\mathbb{X} \right\},\;\mathbf{U}\triangleq \left\{I_n\otimes u^{\intercal}\;|\;u\in\mathbb{U} \right\}$ and  $\widetilde{\Psi}\triangleq  \text{\textbf{co}}\left(\left\{ \psi^{[i]}-\psi^{[j]}\;\Big{|}\;i\neq j\;\; \forall i,j\in\mathbb{I}_1^L \right\} \right) $.\footnote{{\color{black}{The RPI assumption on $\widetilde{\mathbb{X}}_0$ is also later exploited in constraint tightening in Sec. \ref{Seciv}.}}} 
\end{assu}

 The RPI condition on $\widetilde{\mathbb X}_0$ yields $F^{i+1}\widetilde{\mathbb{X}}_0\subseteq F^i\widetilde{\mathbb{X}}_0$ $\forall i\in\mathbb{I}_0^\infty$. This leads to a finite uniform upper bound for the term $||CF^t\tilde{x}_0||$, given by $\Delta_2\triangleq \max_{t\in\mathbb{I}_0^\infty,\;w\in\widetilde{\mathbb{X}}_0} ||CF^t w||=\max_{w\in\widetilde{\mathbb{X}}_0} ||Cw||.$
The finite bounds of $\Delta_1$ and $\Delta_2$ imply that the perturbation-like terms in \eqref{yhattilde} are bounded.

\begin{lemma}\label{L01}
{\color{black}{Suppose the input $u_t$ to the plant \eqref{osys1} and the observer \eqref{adob0} is stabilizing, and the input and output are uniformly bounded, i.e., $||u_t||\leq \bar{u}$ and $||y_t||\leq \bar{y}$ $\forall t\in\mathbb{I}_0^\infty$, respectively, where $\bar u$ and $\bar y$ are finite,}} then, 
the parameter update law \eqref{gd1} ensures
\begin{itemize}
    \item $e_t$, $\Gamma_te_t$, $\tilde{p}_t\in\mathcal{L}_\infty,$ where $e_t\triangleq \Gamma_t^{-2}(y_t-CF^t\hat{x}_0-\phi_t^{\intercal}\hat{p}_{t-1})$ and $\Gamma_t^2\triangleq I_q+\phi_t^{\intercal}\phi_t$,
    \item  $e_t$, $\Gamma_te_t$, $||\hat{p}_t-\hat{p}_{t-1}||\in\mathcal{S}(\Delta_0^2)$, where $\Delta_0$ is the upper bound of $||\Gamma_t^{-1}||||\eta_t+CF^t\tilde{x}_0||$ over all $t\in\mathbb{I}_0^\infty$ \footnote{Any sequence vector $z_t$ is said to belong to $\mathcal{S}(\Delta^2_0)$ if $\sum_{i=t}^{t+k}z^{\intercal}_iz_i\leq c_0\Delta^2_0k+c_1$ $\forall \;t\in\mathbb{I}_1^\infty$, a given constant $\Delta_0^2$, and some $k\in\mathbb{I}_1^\infty$, where $c_0,\;c_1\geq0$ \cite[Theorem~4.11.2, footnote~6]{ioannou2006adaptive}. Here, $\Delta_0\leq \Delta_1+\Delta_2$ (by definitions of $\Delta_1$, $\Delta_2$ and $\Gamma_t$).}, and
    \item the state estimation error $\tilde{x}_t\in\mathcal{L}_\infty$.
\end{itemize}  
\end{lemma}
\begin{proof}
    A sketch of the proof is given in Appendix \hyperlink{App1}{I}.
\end{proof}
{\color{black}{Beginning with an initially feasible COCP (formulated in Sec. \ref{Sec5}), the repeated sequential operation of the COCP followed by the observer guarantees a stabilizing $u_t$, and uniform bounds $\bar{u}\leq \max_{u\in\mathbb{U}}||u||$ on $u_t$ and $\bar y \leq \max_{x\in\mathbb{X}}||Cx||$ on $y_t$, leading to $M_t,\;\phi_t,\;\tilde{x}_t\in\mathcal{L}_\infty$ (as detailed in Remark \ref{remnew}).}} 
Additional guarantees on convergence of $\tilde{x}_t$ are mentioned in Lemma \ref{Lx0}.

\section{Components Required to Design the COCP}\label{Seciv}
{\color{black}{To reformulate the COCP in terms of state and parameter estimates, it is essential to compute suitable constraints for the state estimate trajectory. This involves accounting for the errors that arise when making predictions using the observer dynamics. Additionally, it is important to design a tube structure and a terminal set that allow robust handling of the errors in prediction while providing stability guarantees.}}
\subsection{Constraint for the State Estimate} Similar to \cite{anch}, we begin by rewriting the expression of $\tilde{x}_t$ in  \eqref{xtilde} using the solution of \eqref{FilterM}:
\begin{align}\label{xtil2}
\tilde{x}_{t}=F^t\tilde{x}_0+\textstyle\sum_{k=0}^{t-1}F^{t-1-k}([Y_k\;\;\;U_k]\tilde{p}_t+d_k).
\end{align}
In \cite{anch}, a single invariant set containing all possible values of $\tilde{x}_t$ $\forall t\in\mathbb{I}_0^\infty$ is used to obtain a single constraint set for $\hat{x}_t$. The resulting conservatism in the initial feasible region of the COCP for \cite{anch} is reduced by exploiting  
 \eqref{xtil2}. To this end, we compute time-dependent sets $\widetilde{\mathbb{X}}_t\ni\tilde{x}_t$ for each $t\in\mathbb{I}_0^\infty$ {\color{black}{\cite{mayne2009robust}}} with the help of the recursive relation 
\begin{align}
{\color{black}{\widetilde{\mathbb{X}}_{t+1}\triangleq F \widetilde{\mathbb{X}}_t\oplus (\Omega_1^{yu}\oplus\mathbb{D}) \;\;\;\forall t\in\mathbb{I}_0^\infty.}} \label{Xtildedyn}\end{align}
%
Constraint sets for the state estimates $\hat{x}_t$ are defined as
\begin{align}
    \hat{x}_t= x_t-\tilde{x}_t\in\widehat{\mathbb{X}}_t \triangleq \mathbb{X}\ominus \widetilde{\mathbb{X}}_t \;\; \;\;\forall t\in\mathbb{I}_0^\infty.\label{tight2}
\end{align}
Following Assumption \ref{Awo1} and \eqref{Xtildedyn}, we obtain $\widetilde{\mathbb{X}}_{t+1}\subseteq \widetilde{\mathbb{X}}_t$ $\forall t\in\mathbb{I}_0^\infty$. Consequently, the size of $\widehat{\mathbb{X}}_t$ increases with $t$, i.e., $\widehat{\mathbb{X}}_t\subseteq \widehat{\mathbb{X}}_{t+1}$ $\forall t\in\mathbb{I}_0^\infty$, resulting in a larger region for the state estimate trajectory to evolve, thereby, increasing the initial feasible region of the subsequently reformulated COCP.

\subsection{State Prediction in MPC COCP }
The predictions of the state estimates are obtained using the adaptive observer dynamics rewritten below using \eqref{adob0}:
\begin{align}
    \hat{x}_{t+i+1}&=F\hat{x}_{t+i}+\begin{bmatrix}
         Y_{t+i} &U_{t+i}
     \end{bmatrix}\hat{p}_{t+i} \nonumber\\
     &+\textstyle\sum_{k=0}^{t+i}F^{t+i-k}\begin{bmatrix}
         Y_k& U_k
     \end{bmatrix}(\hat{p}_{t+i+1}-\hat{p}_{t+i})  \label{AOsyspr}\\
     &\;\;\;\;\;\;\;\;\;\;\;\;\forall (t,\;i)\in\mathbb{I}_0^\infty\times \mathbb{I}_0^{N-1},\nonumber
\end{align}
where $t$ is the current time, and $i\in\mathbb{I}_0^{N-1}$ denotes the number of steps ahead of $t$ in the future. Since the terms $Y_{t+i}$ and $\hat{p}_{t+i}$ $\forall i\in\mathbb{I}_1^{N-1}$ in \eqref{AOsyspr} are unavailable at current time $t$, the dynamics in \eqref{AOsyspr} can not be utilized directly for prediction. Motivated by \cite{langson2004robust,rakovic2012homothetic,lorenzen2019robust} to use robust tubes, we divide the dynamics into two parts - one comprising of known quantities, and the other containing the uncertainties, as
\begin{align}
    &\hat{x}_{t+i+1}=F\hat{x}_{t+i}+\begin{bmatrix}
         \widehat{Y}_{t+i} &U_{t+i}
     \end{bmatrix}\hat{p}_{t} +{\color{black}{\varepsilon_{t,i}}}\nonumber\\
     &=\widehat{A}_t \hat{x}_{t+i}+\widehat{B}_tu_{t+i}+\varepsilon_{t,i}\;\;\;\forall (t,\;i)\in\mathbb{I}_0^\infty\times \mathbb{I}_0^{N-1},\label{preddyn}
\end{align}
where $\widehat{Y}_{t+i}\triangleq I_n\otimes (C\hat{x}_{t+i})^\intercal$ is the estimate of $Y_{t+i}$, and $\varepsilon_{t,i}$ $\in\mathbb{R}^n$ contains all the unknown quantities and is termed as the prediction uncertainty. Comparing \eqref{AOsyspr} with \eqref{preddyn}, we can write $\varepsilon_{t,i}\triangleq \begin{bmatrix}
    Y_{t+i}&U_{t+i}
\end{bmatrix}\hat{p}_{t+i}-\begin{bmatrix}
    \widehat{Y}_{t+i}&U_{t+i}
\end{bmatrix}\hat{p}_t +\sum_{k=0}^{t+i}F^{t+i-k}\begin{bmatrix}
         Y_k& U_k
     \end{bmatrix}(\hat{p}_{t+i+1}-\hat{p}_{t+i})$, which after simple manipulation yields
\begin{align}
\varepsilon_{t,i}     
     &= (\widehat{A}_{t+i}-F )\tilde{x}_{t+i}+\textstyle\sum_{k=t}^{t+i-1}\begin{bmatrix}
         \widehat{Y}_{t+i}&U_{t+i}
     \end{bmatrix}(\hat{p}_{k+1}\nonumber\\
     -\hat{p}_{k}&)+\textstyle\sum_{k=0}^{t+i}F^{t+i-k}\begin{bmatrix}
         Y_k& U_k
     \end{bmatrix}(\hat{p}_{t+i+1}-\hat{p}_{t+i}),\label{fstterm}
\end{align}
where $\widetilde{Y}_{t+i}\triangleq Y_{t+i}-\widehat{Y}_{t+i}=I_n \otimes \tilde{y}_{t+i}^\intercal$. For the first term in \eqref{fstterm}, we define the following set:
\begin{align}
   & (\widehat{A}_t-F )\tilde{x}_{t}\in\;\bar{\Omega}_{t}\triangleq  {\{} ( \widehat{A}-F )\tilde{x} \;{|}\; \widehat{A}=\mathcal{M}_A(\hat{a},n), \text{ where} \nonumber\\
    & \hat{a}\in\mathbb{R}^{qn},\; \begin{bmatrix}
        \hat{a}^{\intercal}&\hat{b}^{\intercal}
    \end{bmatrix}^{\intercal}\in\Psi ,\; \tilde{x}\in\widetilde{\mathbb{X}}_{t} {\}}\;\;\forall t\in\mathbb{I}_0^\infty .\label{Obar}
\end{align}
For the second term in \eqref{fstterm} which contains $i$ separate terms in the summation, we define sets $\widehat{\Omega}_{t}$ for the individual terms similar to that of $\Omega^{yu}_1$ in \eqref{Omega_yu_1} as
\begin{align}\label{Omegahat}
    \widehat{\Omega}_{t}\triangleq \{ [
        \widehat{Y}\;\;\;U
    ] \tilde{p}\;{|} \; \widehat{Y}\in\widehat{\mathbf{Y}}_t,\;U\in{\mathbf{U}},\;\tilde{p}\in\widetilde{\Psi} \},
\end{align}
where $\widehat{\mathbf{Y}}_t \triangleq \left\{ I_n \otimes \hat{y}^{\intercal} \big{|}\; \hat{y} =C \hat{x}\; \;\;\forall \hat{x} \in \widehat{\mathbb{X}}_t \right\}$. 
And, the third term in \eqref{fstterm} belongs to the set $\Omega^{yu}_{t+i+1}$, where
\begin{align}
    {\color{black}{\Omega^{yu}_{t+1}\triangleq F\Omega_t^{yu}\oplus \Omega_1^{yu} \;\;\forall t\in\mathbb{I}_1^\infty}}   .\label{Omega_yu_t}
\end{align}
Using \eqref{Obar}-\eqref{Omega_yu_t}, we can write that $\varepsilon_{t,i}$ belongs to
\begin{align}\label{epsilset}
{\color{black}{\mathcal{E}_{t,i}}}\triangleq \bar{\Omega}_{t+i}\oplus i\widehat{\Omega}_{t+i}\oplus{\color{black}{\Omega^{yu}_{t+i+1}}}\;\;\;\forall (t,i)\in\mathbb{I}_0^\infty\times\mathbb{I}_0^{N-1}.
\end{align}

Also, we compute the sets $\widehat \Omega_\infty$ and $\Omega^{yu}$ offline, which are used for designing the tube structure and the terminal set as shown next. {\color{black}{For $\widehat \Omega_\infty$, we require \eqref{Omegahat} and the set $\widehat{\mathbb{X}}_\infty\triangleq \mathbb X \ominus (\Omega^{yu}\oplus \Omega^d)$}}, where $\Omega^{yu}$ is the outer RPI approximation of the minimal RPI set $\Omega^{yu}_\infty$ that satisfies $F\Omega^{yu}_\infty\oplus\Omega^{yu}_1\subseteq\Omega^{yu}_\infty$ \cite{rakovic2005invariant}.   



\subsection{Tubes}
Solving the COCP results in tubes for state estimate and control input. At each time $t$, we obtain a set of $N+1$ tube sections for the observer state 
\begin{align}
    \mathcal{T}^{\hat{x}}_t\triangleq \left\{ \mathbb{T}^{\hat{x}}_{0|t},\; \mathbb{T}^{\hat{x}}_{1|t},\; ...,\;\mathbb{T}^{\hat{x}}_{N|t} \right\},
\end{align}
and a set of $N$ tube sections for control input
\begin{align}
     \mathcal{T}^{u}_t\triangleq \left\{ \mathbb{T}^{u}_{0|t},\; \mathbb{T}^{u}_{1|t},\; ...,\;\mathbb{T}^{u}_{N-1|t} \right\}.
\end{align}
The tube sections in $\mathcal{T}_t^{\hat{x}}$ are designed to be scaled and shifted versions of a user-defined convex polytope $\mathbb{H}\triangleq \text{\textbf{co}} \left ( \left\{h^{[1]},\;h^{[2]},...,h^{[\delta]} \right\} \right)$ with  $\delta$ number of known vertices and containing the origin. For achieving desirable robustness and stability guarantees, the following standard assumption \cite{rakovic2012homothetic} is made on the structure of $\mathbb{H}$ using a robust upper bound of the prediction uncertainty. 
\begin{assu}\label{BasicP}
   The set $\mathbb{H}$ satisfies {\color{black}{$\left( \widehat{A}\oplus\widehat{B}K \right)\mathbb{H}\oplus \Omega_\varepsilon\subseteq \mathbb{H},$}}
   where $(\widehat{A},\widehat{B})\in\mathbb{M}_\Psi$ and {{\color{black}$\Omega_\varepsilon\triangleq \bar{\Omega}_0\oplus (N-1)\widehat{\Omega}_\infty\oplus\Omega^{yu}.$ }}
\end{assu}

Assumption \ref{BasicP} implies that 
$\mathbb{H}$ is RPI to the parametric uncertainties in $\widehat{A}+\widehat{B}K$ and the prediction uncertainties in $\varepsilon$; it can be constructed using the algorithms given in \cite{dey2024computation}, \cite{kouramas}. Since $\mathbb{H}$ is computed offline, $\Omega_\varepsilon$ is defined such that $\Omega_\varepsilon\supseteq \mathcal{E}_{t,i}$ $\forall(t,i)\in\mathbb{I}_0^\infty\times\mathbb{I}_0^{N-1}$.

The tube sections of $\mathcal{T}_t^{\hat{x}}$ are given by
\begin{align}
\mathbb{T}^{\hat{x}}_{i|t}\triangleq  \alpha_{i|t}\oplus \beta_{i|t}\mathbb{H}\subseteq \widehat{\mathbb{X}}_{t+i}\;\;\;\forall i\in\mathbb{I}_0^N ,\label{Tube1}
\end{align}
where $\alpha_{i|t}\in\mathbb{R}^n$ and $\beta_{i|t}\geq 0$ are centers and scaling factors of $\mathbb{T}^{\hat{x}}_{i|t}$, respectively. Using the vertices of $\mathbb{H}$ and \eqref{Tube1}, the tube sections of $\mathcal{T}_t^{\hat{x}}$ are represented as
\begin{align}
  &  \mathbb{T}^{\hat{x}}_{i|t}=\text{\textbf{co}} \left( \left\{ s^{[1]}_{i|t},\;s^{[2]}_{i|t},\;...,s^{[\delta]}_{i|t} \right\} \right), \text{ where} \label{Tube2}\\
  &  s^{[j]}_{i|t}\triangleq \alpha_{i|t}+\beta_{i|t}h^{[j]}\;\;\;\forall (i,j,t)\in\mathbb{I}_0^N\times\mathbb{I}_1^{\delta}\times\mathbb{I}_0^\infty. \label{salphabeta}
\end{align}
Corresponding to \eqref{Tube2}, the control tube sections are written as
\begin{align}
    \mathbb{T}^{u}_{i|t}\triangleq \text{\textbf{co}} \left( \left\{ v^{[1]}_{i|t},\;v^{[2]}_{i|t},\;...,v^{[\delta]}_{i|t} \right\} \right) \subseteq\mathbb{U}&\nonumber\\
    \forall (i,j,t)\in\mathbb{I}_0^{N-1}&\times\mathbb{I}_1^{\delta}\times\mathbb{I}_0^\infty. \label{Tube3}
\end{align}
{{Let the state $\hat{x}_{i|t}$, belonging to the tube section $\mathbb{T}^{\hat{x}}_{i|t}$, be expressed using the convex combination of its vertices as
\begin{align}\label{convX}
    \hat{x}_{i|t}=\textstyle\sum_{j=1}^\delta \tau^{[j]}_{(i,\;t)}s_{i|t}^{[j]},
\end{align}
where $\sum_{j=1}^\delta \tau_{(i,\;t)}^{[j]}=1$ and each $\tau^{[j]}_{(i,\;t)}\in [0,\;1]$. Then, the corresponding control input $u_{i|t}$, to be applied as a function of $\hat{x}_{i|t}$, is given by
\begin{align}\label{convU}
    u_{i|t}(\hat{x}_{i|t})=\textstyle\sum_{j=1}^\delta \tau^{[j]}_{(i,\;t)}v_{i|t}^{[j]}.
\end{align}}}
\subsection{Terminal Set}
The following standard assumption \cite{rakovic2012homothetic},\cite{lorenzen2019robust} and Assumption \ref{ABa} are used to design a suitable terminal set that helps guarantee stability and convergence properties and recursive feasibility.
\begin{assu}\label{AsmTS}
    There exists a non-empty terminal set $\widehat{\mathbb{X}}_{\text{TS}}\subseteq\widehat{\mathbb{X}}_N$ such that $\alpha  \oplus \beta \mathbb{H} \subseteq \widehat{\mathbb{X}}_{\text{TS}}\Rightarrow  \left( \widehat{A}+\widehat{B}K \right)\alpha \oplus (\zeta_1\beta + \zeta_2 ) \mathbb{H}\subseteq \widehat{\mathbb{X}}_{\text{TS}}$, $ K\widehat{\mathbb{X}}_{\text{TS}}\subseteq\mathbb{U}$, where $\alpha\in\widehat{\mathbb{X}}_{\text{TS}},\;\beta\geq 0$, $(\widehat{A},\widehat{B})\in\mathbb{M}_\Psi$,
    \begin{align}
       & \zeta_1\triangleq \min_\zeta \left\{\zeta \;\big{|}\; \left( \widehat{A}+\widehat{B}K \right)\mathbb{H}\subseteq \zeta\mathbb{H} , \; \zeta\in [0,1) \right\} , 
       \label{zeta1}\\
       & \zeta_2\triangleq \min_\zeta \left\{ \zeta \; \big{|} \; \bar{\Omega}_\varepsilon \subseteq \zeta \mathbb{H},\;\zeta\in[0,1] \right\},
       \label{zeta2}\\
& {\color{black}{\bar{\Omega}_\varepsilon\triangleq  \bar{\Omega}_{N-1}\oplus  (N-1)\widehat{\Omega}_\infty\oplus \Omega^{yu} }}.
\label{OmegaTS}
    \end{align}
\end{assu}

The set $\widehat{\mathbb X}_{\text{TS}}$ can be constructed using the algorithm in \cite{dey2024computation}. By definition, $\bar{\Omega}_\varepsilon$ contains all possible values of $\varepsilon_{t,N-1}\in \mathcal{E}_{t,N-1}$ $\forall t\in\mathbb{I}_0^\infty$.  {\color{black}{The construction of $\widehat{\mathbb X}_\text{TS}$ is simplified by considering the sets $\zeta_1\mathbb{H}\supseteq (\widehat{A}+\widehat{B}K)\mathbb{H}$ and $\zeta_2\mathbb{H}\supseteq \bar{\Omega}_\varepsilon$. The constants $\zeta_1$ and $\zeta_2$ satisfy their respective ranges due to Schur stable property of the matrices $( \widehat{A}+\widehat{B}K )$ and by design of $\mathbb{H}$ since $\bar{\Omega}_\varepsilon\subseteq\Omega_\varepsilon$, respectively.}}

\section{
Reformulated COCP for RAOFMPC}\label{Sec5}
The COCP is reformulated as follows with decision variable $\theta_t\triangleq \left\{ \{\alpha_{i|t}\}_{i=0:N},\;\{\beta_{i|t}\}_{i=0:N},\; \left\{v^{[j]}_{i|t} \right\}_{i=0:N-1,\;j=1:\delta} \right\}$.
\begin{align}
&\text{COCP}:\;\;\min_{\theta_t} J(\hat{x}_t,\theta_t),\;\;\text{where }\nonumber\\
& J(\hat{x}_t,\theta_t)\triangleq \sum_{i=0}^{N-1} \sum_{j=1}^\delta\Big{\{}\Big{|}\Big{|}s^{[j]}_{i|t}\Big{|}\Big{|}^2_Q+\Big{|}\Big{|}v^{[j]}_{i|t}\Big{|}\Big{|}^2_R  \Big{\}}  +\sum_{j=1}^{\delta}\Big{|}\Big{|}s^{[j]}_{N|t}\Big{|}\Big{|}^2_P \label{MPC2}\\
&\text{subject to {\color{black}{\eqref{Tube1}-\eqref{Tube3}}},} \nonumber\\
&\alpha_{0|t}=\hat{x}_t,\;\beta_{0|t}=0\text{ and }\beta_{i|t}\geq 0\;\;\;\forall i\in\mathbb{I}_{1}^{N},  \tag{\ref{MPC2}a}\label{cons1m}\\
   & \alpha_{N|t}=  ( \widehat{A}_t+\widehat{B}_tK ) \alpha_{N|t-1},\;\; \;t\in\mathbb{I}_1^\infty,  \tag{\ref{MPC2}b}\label{cons7m} \\
   & \beta_{N|t}\leq \zeta_1\beta_{N|t-1}+\zeta_2, \;\;\;t\in\mathbb{I}_1^\infty , \tag{\ref{MPC2}c}\label{cons8m} \\ 
  &\mathbb{T}^{\hat{x}}_{i|t}\subseteq \widehat{\mathbb{X}}_{t+i},\text{ }\mathbb{T}^u_{i|t}\subseteq\mathbb{U}\;\;\;\forall i\in\mathbb{I}_{0}^{N-1} , \tag{\ref{MPC2}d}\label{cons2m}\\
   & \mathbb{T}^{\hat{x}}_{N|t}\subseteq \widehat{\mathbb{X}}_{\text{TS}}, \text{ and}  \tag{\ref{MPC2}e}\label{cons3m}\\
    &\widehat{A}_ts_{i|t}^{[j]}+\widehat{B}_tv_{i|t}^{[j]}\in \mathbb{T}^{\hat{x}}_{i+1|t}\ominus \mathcal{E}_{t,i} \; \forall (i,j)\in\mathbb{I}_{0}^{N-1}\times\mathbb{I}_1^\delta.  \tag{\ref{MPC2}f}\label{cons6m}
\end{align} 

Two additional constraints \eqref{cons7m} and \eqref{cons8m} have been introduced in the COCP for guaranteeing desirable properties of recursive feasibility, stability and convergence. Provided Assumption \ref{AsmTS} holds, the constraints \eqref{cons7m} and \eqref{cons8m} will also be satisfied $\forall t\in\mathbb{I}_1^\infty$. The output of the optimizer is the optimal $\theta^*_t$ that is used to construct the optimal tubes for state estimate and control input at time $t$. Adding the sets $\widetilde{\mathbb{X}}_t$ to the tube for state estimate, obtained from \eqref{MPC2}, yields a robust outer tube for the true state trajectory. This results in a two-tube structure: the inner one containing the trajectory of $\hat{x}_t$ and the outer one containing the trajectory of $x_t$. Algorithm \ref{alg} provides the steps to implement the RAOFMPC framework.

\begin{algorithm}[h!]
\caption{RAOFMPC}\label{alg}
\begin{algorithmic}[1]
\REQUIRE $\mathbb{X}$, $\mathbb{U}$, $\mathbb{D}$, $\psi^{[i]}$ $\forall i\in\mathbb{I}_1^L$, $F$, $\hat{\psi}_0$, $\widetilde{\mathbb{X}}_0$, $\hat{x}_0$, $N$, $Q$, $R$, $\kappa$.
\ENSURE  $\theta_t^*$ $\;\forall t\in\mathbb{I}_0^\infty$.\\
\STATE \textbf{Offline Steps:} Compute {\color{black}{$\widetilde{\mathbb{X}}_i$ $\forall i\in\mathbb I_1^N$, ${\widehat{\mathbb{X}}_i}$ $\forall i\in\mathbb{I}_0^N$, $\Omega_i^{yu}$ $\forall i\in\mathbb{I}_1^{N}$, $\bar{\Omega}_{i}$, $\widehat{\Omega}_i$ and $\mathcal{E}_{0,i}$ $\forall i\in\mathbb{I}_0^{N-1}$}} using \eqref{Xtildedyn}, \eqref{tight2}, \eqref{Obar}-\eqref{epsilset}.
Also, compute 
$\mathbb{H}$
following Assumption \ref{BasicP},
$(P,K)$ following Assumption \ref{ABa}, and $\widehat{\mathbb{X}}_\text{TS}$, $\zeta_1$ and $\zeta_2$ following  Assumption \ref{AsmTS}. 
  
{{\textbf{Online Steps:}}}
\STATE Measure $y_0$ as the output of the plant \eqref{osys1} at $t=0$.
\FOR{$t\geq 0$} 
  \STATE   Run the COCP \eqref{MPC2} with $\widehat{A}_t$, $\widehat{B}_t$ and $\hat{x}_t$ to obtain $\theta_t^*$.
 \STATE   Apply $u_t=u_{0|t}^*={v^{[j]}_{0|t}}^*$ (choose any $j\in\mathbb{I}_1^\delta$) to the plant \eqref{osys1}, and measure $y_{t+1}$.
\STATE Apply $u_t$, $y_t$, $y_{t+1}$ to the observer to get $\widehat{A}_{t+1}$, $\widehat{B}_{t+1}$ and $\hat{x}_{t+1}$ using \eqref{FilterM}, \eqref{gd1} and \eqref{xhat}. 
   \STATE  Update $t \gets t+1$.
   \STATE {\color{black}{Compute $\widetilde{\mathbb{X}}_{t+N}$, $\widehat{\mathbb{X}}_{t+N}$, $\bar{\Omega}_{t+N-1}$, $\widehat{\Omega}_{t+N-1}$, $\Omega^{yu}_{t+N}$, $\mathcal{E}_{t,i}$ $\forall i\in\mathbb{I}_0^{N-1}$ using \eqref{Xtildedyn}, \eqref{tight2}, \eqref{Obar}-\eqref{epsilset}, respectively}}\footnotemark. 
\ENDFOR
\end{algorithmic}
\end{algorithm}\footnotetext{{\color{black}{$\widetilde{\mathbb{X}}_{t-1}$ and $\Omega^{yu}_{t}$ if $t\neq 1$, and, $\widehat{\mathbb{X}}_{t-1}$, $\bar{\Omega}_{t-1}$, $\widehat{\Omega}_{t-1}$, $\mathcal{E}_{t-1,i}$ $\forall i\in\mathbb{I}_0^{N-1}$ can be removed from memory at this step.}}}

\subsection{Properties of the COCP}
While $\tilde{x}_t$ is robustly handled using constraint tightening by $\widetilde{\mathbb{X}}_t$, the future values of $\varepsilon_{t,i}$ depend on the solution of the optimization problem \eqref{MPC2}. The following lemma shows that the predicted values of $\varepsilon_{t,i}\in\mathcal{E}_{t,i}$ $\forall (t,i)\in \mathbb{I}_0^\infty\times\mathbb{I}_0^{N-1}$.

\begin{lemma}\label{Lrf1}
If ${x}_t\in\mathbb{X}$, $\hat{x}_t\in\widehat{\mathbb{X}}_t$ and the COCP \eqref{MPC2} is feasible at time $t\in\mathbb{I}_0^\infty$ resulting in $\theta_t$, then $\varepsilon_{t,i}\in\mathcal{E}_{t,i}$ $\forall i\in\mathbb{I}_0^{N-1}$, and $\hat{x}_{t+i}\in\widehat{\mathbb{X}}_{t+i}$ $\forall i\in\mathbb{I}_1^N$ guaranteeing $x_{t+i}\in\mathbb{X}$ $\forall i\in\mathbb{I}_1^N$.
\end{lemma}

\begin{proof}
We first show that the statement is true at $t=0$. If ${x}_0\in\mathbb{X}$, $\hat{x}_0\in\widehat{\mathbb{X}}_0$, then $\tilde{x}_0\in\tilde{\mathbb{X}}_0$ (from \eqref{tight2}). And, $\varepsilon_{0,0}=( \widehat{A}_0-F  )\tilde{x}_0+\begin{bmatrix}
Y_0&U_0
\end{bmatrix}(\hat{p}_1-\hat{p}_0)\in\bar{\Omega}_0\oplus\Omega_1^{yu}=\mathcal{E}_{0,0}$. Given the COCP is feasible at $t=0$, we obtain $N$ number of inputs. Implementing $u_0=u_{0|0}=v^{[j]}_{0|0}$ for any $j\in\mathbb{I}_1^\delta$ to the plant \eqref{osys1} and observer \eqref{preddyn} results in $x_1$ and $\hat{x}_1$, respectively, where \eqref{cons2m} and \eqref{cons6m} guarantee that $\hat{x}_1\in\mathbb{T}^{\hat{x}}_{1|0}\subseteq \widehat{\mathbb{X}}_1$. Further, from \eqref{xtil2} and \eqref{Xtildedyn}, we have $\tilde{x}_1\in\tilde{\mathbb{X}}_1$. Therefore, $x_1=\hat{x}_1+\tilde{x}_1\in\mathbb{X}$ (from \eqref{tight2}). The input $u_{i|0}$ to be applied at each step to obtain $\hat{x}_{i+1}$ can be written using a convex combination (see \eqref{convX} and \eqref{convU}). Then, proceeding similarly as done for $\varepsilon_{0,0}$, $\hat{x}_1$ and $x_1$ using the definitions of the sets, it can be shown that $\varepsilon_{0,i}\in \mathcal{E}_{0,i} $ $\forall i\in\mathbb{I}_1^{N-1}$, and $\hat{x}_{i}\in\widehat{\mathbb{X}}_i$ guaranteeing $x_{i}\in\mathbb{X}$ $\forall i\in\mathbb{I}_2^N$. The remaining proof for any $t\in\mathbb{I}_1^\infty$ easily follows by repeating similar steps using the corresponding set definitions.
\end{proof}

\subsubsection{Recursive feasibility}
\begin{theo}\label{theo1}
If Assumptions \ref{ABa}-\ref{AsmTS} hold and the COCP \eqref{MPC2} is feasible at any time $t\in\mathbb{I}_0^\infty$, then it is also feasible at all time $t+i$ $\;\forall i \in \mathbb{I}_1^\infty$.
\end{theo}
\begin{proof}
We begin by proposing a feasible solution of the COCP at time $t+1$ using the optimal $\theta_t^*$ obtained by solving \eqref{MPC2} at time $t$. Let the state estimate and control tubes, respectively, be 
\begin{align}
&  \mathcal{T}^{\hat{x}}_{t+1}= \left\{ \mathbb{T}^{\hat{x}}_{0|t+1}, \mathbb{T}^{\hat{x}}_{1|t+1},...,\mathbb{T}^{\hat{x}}_{N|t+1} \right\}, \text{ and}\label{rf1}\\
&  \mathcal{T}^u_{t+1}= \left\{ \mathbb{T}^u_{0|t+1}, \mathbb{T}^u_{1|t+1},...,\mathbb{T}^u_{N-1|t+1} \right\},\text{ where} \label{rf2}\\
&\mathbb{T}^{\hat{x}}_{i|t+1}={\mathbb{T}^{\hat{x}}_{i+1|t}}^* \;\;\forall i\in\mathbb{I}_1^{N-1}, \;
\label{rf3}\\
&
\mathbb{T}^u_{i|t+1}={\mathbb{T}^{u}_{i+1|t}}^*\;\;\forall i\in\mathbb{I}_1^{N-2}, \label{rf3a}\\
&\mathbb{T}^{\hat{x}}_{N|t+1}= ( \widehat{A}_{t+1}+\widehat{B}_{t+1}K ) \alpha_{N|t}^{*} \nonumber\\
&\;\;\;\;\;\;\;\;\;\;\;\;\;\;\oplus \left( ( \widehat{A}_{t+1}+\widehat{B}_{t+1}K ) \beta_{N|t}^*+\zeta_2 \right)\mathbb{H} \subseteq \widehat{\mathbb{X}}_\text{TS},  \label{rf3b}\\
&\mathbb{T}^u_{N-1|t+1}=K \mathbb{T}^{\hat{x}}_{N-1|t+1}=K{\mathbb{T}^{\hat{x}}_{N|t}}^* \subseteq K\widehat{\mathbb{X}}_{\text{TS}} \subseteq\mathbb{U} , \label{rf4}\\
&\mathbb{T}^{\hat{x}}_{0|t+1}=\{\hat{x}_{t+1}\}\Rightarrow \alpha_{0|t+1}=\hat{x}_{t+1}\text{ and }\beta_{0|t+1}=0. \label{rf5}
\end{align}
    
Since the COCP is feasible at time $t$, the control $u_{0|t}$ will drive the state $\hat{x}_{t+1}$ to belong to the set ${\mathbb{T}_{1|t}^{\hat{x}}}^*$. Let $\tau^{[1]}_{(1,\;t)},\tau^{[2]}_{(1,\;t)},...,\tau^{[\delta]}_{(1,\;t)}$, where $\sum_{j=1}^\delta\tau^{[j]}_{(1,\;t)}=1$ and $0\leq \tau^{[j]}_{(1,\;t)}\leq 1\;\;\forall j\in\mathbb{I}_1^\delta$, be the convex combination satisfying
\begin{align}
    \hat{x}_{t+1}=\hat{x}_{1|t}=\textstyle\sum_{j=1}^\delta \tau^{[j]}_{(1,\;t)} {s^{[j]}_{1|t}}^*. \label{rf6}
\end{align}
Then, the control tube section $\mathbb{T}_{0|t+1}^u$ is proposed as
\begin{align}
    \mathbb{T}^u_{0|t+1}=\{u_{t+1}\}=\{u_{1|t}\}=\{\textstyle\sum_{j=1}^\delta \tau^{[j]} _{(1,\;t)}{v^{[j]}_{1|t}}^*\}. \label{rf7}
\end{align}

Following \eqref{rf3a}-\eqref{rf7}, the vertices of the tube sections are
\begin{align}
    & s^{[j]}_{0|t+1}=\hat{x}_{t+1}=\textstyle\sum_{k=1}^\delta \tau^{[k]}_{(1,\;t)} {s^{[k]}_{1|t}}^*\;\;\forall j\in\mathbb{I}_1^\delta, \label{RecF1}\\
    & v^{[j]}_{0|t+1}=\hat{u}_{t+1}=\textstyle\sum_{k=1}^\delta \tau^{[k]}_{(1,\;t)} {v^{[k]}_{1|t}}^*\;\;\forall j\in\mathbb{I}_1^\delta, \label{RecF2}\\
    & s^{[j]}_{i|t+1}={s^{[j]}_{i+1|t}}^*\;\;\forall (i,\;j)\in\mathbb{I}_1^{N-1}\times \mathbb{I}_1^\delta, \label{RecF3}\\
    &v^{[j]}_{i|t+1}={v^{[j]}_{i+1|t}}^*\;\;\forall (i,\;j)\in\mathbb{I}_1^{N-2}\times \mathbb{I}_1^\delta, \label{RecF4}\\
    &s^{[j]}_{N|t+1}=(\widehat{A}_{t+1}+\widehat{B}_{t+1}K )\alpha_{N|t}^*+ \nonumber\\
    & \;\;\;\;\;\;\;\left( (\widehat{A}_{t+1}+\widehat{B}_{t+1}K )\beta_{N|t}^* +\zeta_2 \right)h^{[j]}\;\;\forall j\in\mathbb{I}_1^\delta,\\
    &v^{[j]}_{N-1|t+1}= K{s^{[j]}_{N|t}}^*\;\;\forall j\in\mathbb{I}_1^\delta.\label{RecF5}
\end{align}

The conditions \eqref{Tube1}-\eqref{Tube3}, \eqref{cons1m}-\eqref{cons3m} hold at time $t+1$ based on the construction in \eqref{rf1}-\eqref{rf7} and the properties of $\widehat{\mathbb{X}}_{\text{TS}}$ in Assumption \ref{AsmTS}. For proving satisfaction of \eqref{cons6m}, we start with its left-hand side at time $t+1$ for $i=0$ and $\forall j\in\mathbb{I}_1^\delta$. Simple manipulations using \eqref{RecF1} and \eqref{RecF2} allow us to write 
\begin{align*}
 & \;\widehat{A}_{t+1}s_{0|t+1}^{[j]}+\widehat{B}_{t+1}v_{0|t+1}^{[j]} \\
  = &\;\widehat{A}_{t+1} {\hat{x}_{t+1}} +\widehat{B}_{t+1}  {u_{t+1}} \\
  = &\;\widehat{A}_{t}{\hat{x}_{t+1}} +\widehat{B}_{t}  {u_{t+1}} + \begin{bmatrix}
\widehat{Y}_{t+1} & U_{t+1}
\end{bmatrix}(\hat{p}_{t+1}-\hat{p}_t)\\
= &\;\sum_{k=1}^\delta \tau^{[k]}_{(1,\;t)}\left( \widehat{A}_{t} {s_{1|t}^{[k]}}^*+\widehat{B}_{t} {v^{[k]}_{1|t}}^* \right)\\
& \;\;\;\;\;\;\;\;\;\;\;\;\;\;\;\;\;+\begin{bmatrix}    \widehat{Y}_{t+1} & U_{t+1}
\end{bmatrix}(\hat{p}_{t+1}-\hat{p}_t)\\
 \in &\; {\mathbb{T}^{\hat{x}}_{2|t}}^* \ominus\mathcal{E}_{t,1}\oplus\widehat{\Omega}_{t+1} \\
=&\;{\mathbb{T}^{\hat{x}}_{2|t}}^*\ominus\{ \bar{\Omega}_{t+1}\oplus  \widehat{\Omega}_{t+1} \oplus \Omega^{yu}_{t+2}\}\oplus\widehat{\Omega}_{t+1} \\
\subseteq & \;{\mathbb{T}^{\hat{x}}_{2|t}}^*\ominus\{ \bar{\Omega}_{t+1} \oplus \Omega^{yu}_{t+2}\}\\
=& \;\mathbb{T}^{\hat{x}}_{1|t+1}\ominus\mathcal{E}_{t+1,0}.  
\end{align*}


Similarly, using \eqref{RecF3} and \eqref{RecF4} at time $t+1$ $\forall (i,j)\in\mathbb{I}_1^{N-2}\times\mathbb{I}_1^\delta$, we see that 
\begin{align*}
    &\;\widehat{A}_{t+1}s_{i|t+1}^{[j]}+\widehat{B}_{t+1}v_{i|t+1}^{[j]}\\
    =&\; \widehat{A}_{t}{s_{i+1|t}^{[j]}}^*+\widehat{B}_{t} {v^{[j]}_{i+1|t}}^*+\begin{bmatrix}    {\widehat{Y}^{[j]}_{i+1|t}}\;^* & {U_{i+1|t}^{[j]}}^*
\end{bmatrix}(\hat{p}_{t+1}-\hat{p}_t) \\
\in & \;{\mathbb{T}^{\hat{x}}_{i+2|t}}^* \ominus\mathcal{E}_{t,i+1}\oplus\widehat{\Omega}_{t+i+1}\\
\subseteq &\; \mathbb{T}^{\hat{x}}_{i+1|t+1}\ominus\mathcal{E}_{t+1,i},
\end{align*}
where ${\widehat{Y}^{[j]}_{i+1|t}}\;^* \triangleq I_n\otimes \left( C{s_{i+1|t}^{[j]} }^*\right)^{\intercal}$ and ${U_{i+1|t}^{[j]}}^*\triangleq I_n\otimes \left({v^{[j]}_{i+1|t}}^*\right)^{\intercal}$. 

Finally, using \eqref{RecF3} and \eqref{RecF5} for $i=N-1$ and $\forall j$ $ \in \mathbb{I}_1^\delta$, we can write 
\begin{align*}
    &\; \widehat{A}_{t+1}s_{N-1|t+1}^{[j]}+\widehat{B}_{t+1}v_{N-1|t+1}^{[j]}\\
     =& \;( \widehat{A}_{t+1}+\widehat{B}_{t+1}K ) {s_{N|t}^{[j]}}^*\\
    \in &\; ( \widehat{A}_{t+1}+\widehat{B}_{t+1}K ) \alpha_{N|t}^*  \oplus  ( \widehat{A}_{t+1}+\widehat{B}_{t+1}K ) \beta_{N|t}^*\mathbb{H}\\
    \subseteq &\; \mathbb{T}_{N|t+1}^{\hat{x}}\ominus\zeta_2\mathbb{H}\\
    \subseteq  &\; \mathbb{T}_{N|t+1}^{\hat{x}}\ominus \mathcal{E}_{t+1,N-1}\;(\because\mathcal{E}_{t+1,N-1}    \subseteq  \zeta_2\mathbb{H}).
\end{align*}

These above three results together prove that \eqref{cons6m} is satisfied at time $t+1$. Therefore, the solution proposed in \eqref{rf1}-\eqref{RecF5} is feasible at $t+1$. 

Using the method of mathematical induction, it can be similarly proved that solutions exist at $t+2$, $t+3$, and so on. Therefore, the COCP is recursively feasible.
\end{proof}
\subsubsection{Stability and Boundedness}\label{stabbound}
In addition to the parametric uncertainty, the system dynamics in \eqref{osys1} contains an added disturbance, while the dynamics \eqref{adob0} or \eqref{preddyn} used to predict the future observer states contains prediction uncertainty that acts as a disturbance-like term. Convergence of $x_t$ and $\hat{x}_t$ to the respective origins of \eqref{osys1} and \eqref{adob0} cannot be guaranteed owing to the presence of the disturbances. Instead, we prove robust exponential stability \cite[Def.~2]{mayne2006robust}, i.e., the trajectories of the state and state estimate converge exponentially to sets dependent on their corresponding uncertainty uniform bounds. 

\begin{lemma}\label{Lx0}
    If the COCP \eqref{MPC2} is recursively feasible and Assumptions \ref{ABa}-\ref{AsmTS} hold, then $\tilde{x}_t$ exponentially converges to an element in the set $\Omega_\infty^{yu}\oplus \Omega_\infty^d$. Additionally, if $\tilde{p}_t\rightarrow 0$ as $t\rightarrow\infty$, then $\tilde{x}_t$ exponentially converges to the set $\Omega_\infty^d$.
\end{lemma}
\begin{proof}
The proof directly follows from \eqref{xtil2} by exploiting the Schur stable property of $F$.
\end{proof}

\begin{theo}\label{theoSTability}
Suppose Assumptions \ref{ABa}-\ref{AsmTS} hold, $x_0\in\mathbb{X}$, $\hat{x}_0\in\widehat{\mathbb{X}}_0$ and the COCP \eqref{MPC2} is feasible at $t=0$. Then, by virtue of Lemmas \ref{L01}-\ref{Lx0} and Theorem \ref{theo1}, all the signals are bounded $\forall t\in\mathbb{I}_0^\infty$, and the dynamics in \eqref{osys1} {{and \eqref{adob0} exhibit robust exponential stability \cite[Def.~2]{mayne2006robust}}}. 
\end{theo}
\begin{proof}
Since the COCP is initially feasible, combining Lemma \ref{Lrf1} and Theorem \ref{theo1} and using the method of mathematical induction, it can be shown that $\varepsilon_{t,i}\in \mathcal{E}_{t,i} $ $\forall (t,i)\in\mathbb{I}_0^\infty\times \mathbb{I}_0^{N-1}$, $\tilde{x}_{t}\in\widetilde{\mathbb{X}}_t $, $\hat{x}_{t}\in\widehat{\mathbb{X}}_t $ and $x_{t}\in\mathbb{X}$ $\forall t\in\mathbb{I}_0^\infty$. Also, Theorem \ref{theo1} guarantees recursive feasibility that implies that $u_t\in\mathbb{U}$ $\forall t\in\mathbb{I}_0^\infty$. The projection operator in \eqref{gd1} ensures that $\hat{p}_t\in\Pi$ $\forall t\in\mathbb{I}_0^\infty$. By definition, $M_t,\;\phi_t\in\mathcal{L}_\infty$ (see Remark \ref{remnew}). Therefore, we can write, $x_t$, $\hat{x}_t$, $\tilde{x}_t$, $u_t$, $\varepsilon_{t,i}$, $\hat{p}_t\in\mathcal{L}_\infty$.

Given the COCP is initially feasible, its recursive feasibility implies that the constraints \eqref{cons7m} and \eqref{cons8m} are satisfied $\forall t\in\mathbb{I}_1^\infty$. Once the state estimate trajectory enters $\widehat{\mathbb{X}}_{\text{TS}}$, the quadratically stabilizing gain $K$ takes over. Thereafter,  following \eqref{cons7m}, the center of the tube sections exponentially converges to the origin as $t\rightarrow \infty$ with a convergence rate ${\lambda}_\alpha \in[\min_{(\widehat{A},\widehat{B})\in\mathbb{M}_\Psi}\lambda(\widehat{A}+\widehat{B}K), \;\max_{(\widehat{A},\widehat{B})\in\mathbb{M}_\Psi}\lambda(\widehat{A}+\widehat{B}K)]\subseteq(-1,1)$.

Further, following the dynamics in \eqref{cons8m}, the scaling factor of the tube cross-sections exponentially converges to a non-negative scalar $\bar \beta$ as $t\rightarrow \infty$ with a convergence rate between $0$ and $\zeta_1$, where
\begin{align}
{\color{black}{\bar \beta\leq(1-\zeta_1)^{-1}\zeta_2}}. 
\end{align} 
The tube cross-sections for $\hat{x}_t$, therefore, converge to $\bar{\beta}\mathbb{H}$ as $t\rightarrow \infty$. The exponential convergence of the tube center and the scaling factor is ensured by design from \eqref{Asmeq} and Assump-\\ tion \ref{AsmTS} with \eqref{zeta1}, respectively. 
Using Lemma \ref{Lx0}, we can conclude that the outer-tube cross-section exponentially converges to $\bar{\beta} \mathbb{H}\;\oplus\;\Omega_\infty^{yu}\oplus\;\Omega_\infty^d$ as $t\rightarrow\infty$. In addition, if $\tilde{p}_t\rightarrow 0$, then the outer-tube cross-section exponentially converges to the set $\bar{\beta} \mathbb{H}\;\oplus\;\Omega_\infty^d$ as $t\rightarrow\infty$. Therefore, the dynamics in \eqref{osys1} and \eqref{adob0} exhibit robust exponential stability.
 \end{proof}
 
\begin{remark}\label{remnew}
{\color{black}{The online phase of the algorithm begins with solving the COCP, as outlined in Step 4,}} with a choice of $\widehat{A}_0$, $\widehat{B}_0$ and $\hat{x}_0$, all of which are bounded. If the COCP is initially feasible, we get a stabilizing $u_0$, and uniformly bounded $u_0$ and $y_1$. {\color{black}{Then, $u_0\leq \bar u$ and $y_0$, $y_1\leq \bar y$ are used by the observer (Step 6)}}. Also, $M_1$ and $\phi_1$ are bounded (by definition). As a result, the estimates $\widehat{A}_1$, $\widehat{B}_1$, $\hat{x}_1$ are bounded (Lemma \ref{L01}, \eqref{xhat}), implying $\tilde x_1$ in \eqref{xtilde} is bounded. {\color{black}{The new estimates are used to solve the COCP again (Step 4). Repeating the same analysis at each time instant in this sequential manner of COCP followed by observer,}} it can be guaranteed that $M_t$, $\phi_t,\;\tilde{x}_t\in\mathcal{L}_\infty$.
\end{remark}
\section{Numerical Example}
Consider the following MIMO LTI system having $4$ states, $2$ inputs and $2$ outputs 
{{\begin{align*}
   &x_{t+1}= \text{\tiny{$\begin{bmatrix}-2.1011   & 0.0065&1&0\\ 0 &  -0.005 &0 & 1\\-0.001  &  0.3019   &      0   &      0   \\ 1.0189  &  0.2    &     0    &     0
   \end{bmatrix}$}}x_t+\text{\tiny{$\begin{bmatrix}12 & 10\\
   3.011 & 0\\
   -0.11 & -1.23\\
   -6.031 & 0.041
   \end{bmatrix}$}}u_t+d_t,\\
   &y_t=\text{\tiny{$\begin{bmatrix}
       x_1^{(1)}&x_t^{(2)}
   \end{bmatrix}^\intercal$}},
\end{align*}}}
with {\color{black}{$||d_t||_\infty \leq 2$}}, {\color{black}{$x_0=\begin{bmatrix}
    -40&-33&32.95&35.96
\end{bmatrix}^\intercal$}}, and hard constraints $||x_t||_\infty\leq 40$, $||u_t||_\infty\leq 4$. The set $\Psi$ has $5$ vertices\footnote{\tiny{$\psi^{[1]}=[ -2.101 \;\;0.0064 \;\;0\;\; -0.005\;\;-0.001\;\;0.3019\;\;1.0189\;\;0.2 \;\;12\;\;10\;\; 3.0101$ $0 \;\; -0.11\;\;-1.23 \;\; -6.03\;\;0.041]^{\intercal}$, $\psi^{[2]}=[-2.102 \;\;0.007\;\; 0\;\; -0.005\;\;-0.001$ $0.3019\;\; 1.0189\;\;0.2\;\;12\;\;10 \;\; 3.0201 \;\;0\;\; -0.11 \;\;-1.23\;\;-6.04\;\; 0.041]^{\intercal}$, $\psi^{[3]}=[-2.103$ $0.0059\;\; 0\;\; -0.005\;\;-0.001\;\;0.3019\;\; 1.0189\;\;0.2\;\;12\;\;10 \;\; 3.0301 \;\;0\;\; -0.11 $ $-1.23\;\;-6.05\;\; 0.041]^{\intercal}$, $\psi^{[4]}=[ -2.100 \;\;0.0064 \;\;0\;\; -0.005\;\;-0.001\;\;0.3019\;\;1.0189$\\$0.2 \;\;12\;\;10\;\; 3.0101\;\;0 \;\; -0.11\;\;-1.23 \;\; -6.025\;\;0.041]^{\intercal}$, $\psi^{[5]}=[-2.103 \;\;0.0072\;\; 0$ $ -0.005\;\;-0.001\;\;0.3019\;\; 1.0189\;\;0.2\;\;12\;\;10 \;\; 3.0301 \;\;0\;\; -0.11$ $ -1.23\;\;-6.02\;\; 0.041]^{\intercal}$.}}
with uncertainties in the 1$^\text{st}$, 2$^\text{nd}$, 11$^\text{th}$ and 15$^\text{th}$ entries of the parameter vector $\psi$. We consider $F=[-0.001\; 0.0064$\\$\;1\; 0;\; 0 $ $-0.005\;0\;1;\;-0.001\; 0.02\;0\;0;\; 0.02\; 0.2001\;0\;0]$, $Q=2I_4$, $R=I_2$,  $K=[0.148\;0.0225 \;-0.0028 \; -0.042;\;0.0263 $ $-0.0265 \;-0.0958 \;0.0383 ]$, {\color{black}{$N=7$}}, $\widetilde{\mathbb{X}}_0$ is the maximal RPI set in $1/5$ of $\mathbb{X}$, $\hat{x}_0=[-32\;-25\;\;27\;\;30]^\intercal$, $\hat \psi_0=0.1\psi^{[1]}+0.9\psi^{[3]}$,\\ $\zeta_1=0.98$, $\zeta_2=0.28$ and $\kappa=0.05$. Fig.s \ref{STube} (a) and (b) demonstrate stabilization performance and constraint satisfaction at all times, and Fig. \ref{STube} (c)
shows that $||\tilde{p}_t||_2$ is bounded.
\footnote{{\color{black}{Total offline computation time was 182.61 seconds, whereas the total time for the online set computation and COCP to run once was 0.927 seconds on average (over 15 time instants), with the use of MATLAB R2021b and toolboxes \cite{MPT3} and \cite{Lofberg2004} in an AMD Ryzen 9 5900HX 330GHz processor, 16 GB RAM, 64-bit operating system}}.} 

\begin{figure}[t!]
\vspace{0.23cm}\centering
\framebox{\parbox{3in}{\includegraphics[scale=0.425]{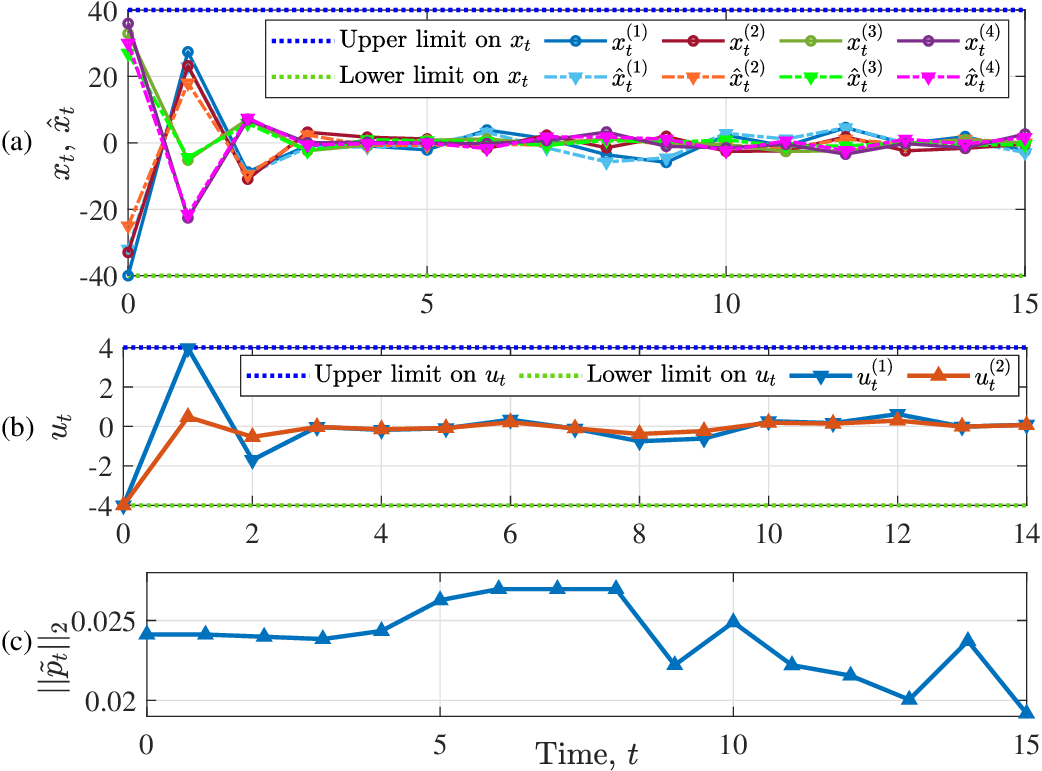}}}
\caption{(a) State and state estimate trajectories, (b) Control input, and (c) 2-norm of parameter estimation error.}       \label{STube}
\end{figure}
\begin{figure}[!t]
\centering
\framebox{\parbox{3in}{\includegraphics[scale=0.216]{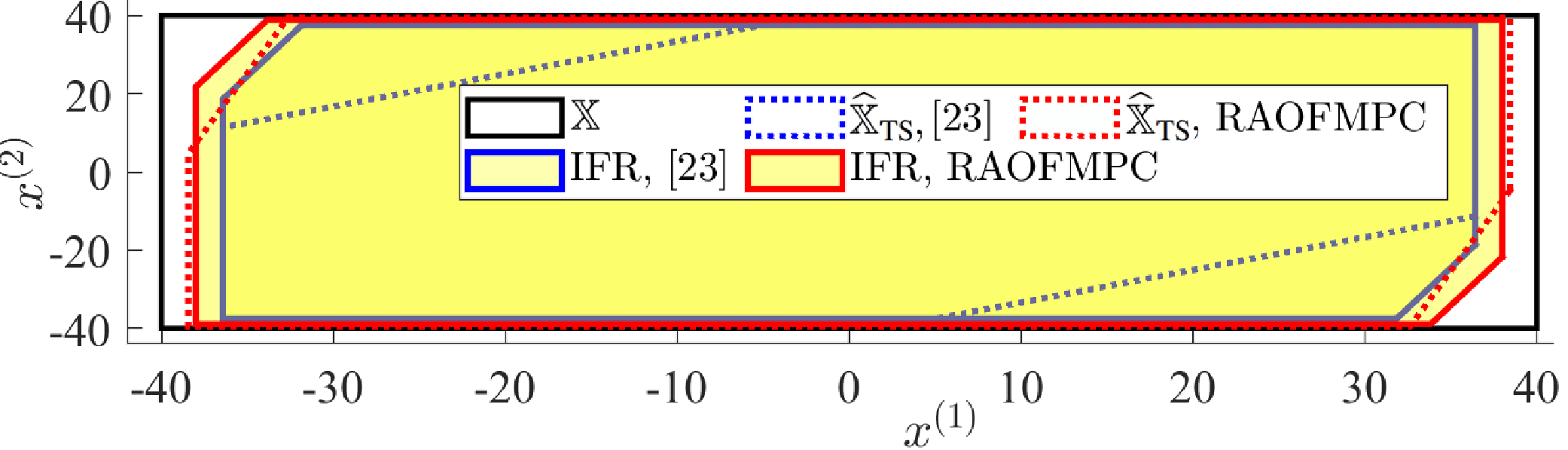}}}
\caption{Initial feasible regions (IFR) for \cite{anch} and the proposed method.} 
\label{sisoset}
\end{figure}
\begin{figure}[!t]
\centering
\framebox{\parbox{3in}{\includegraphics[scale=0.425]{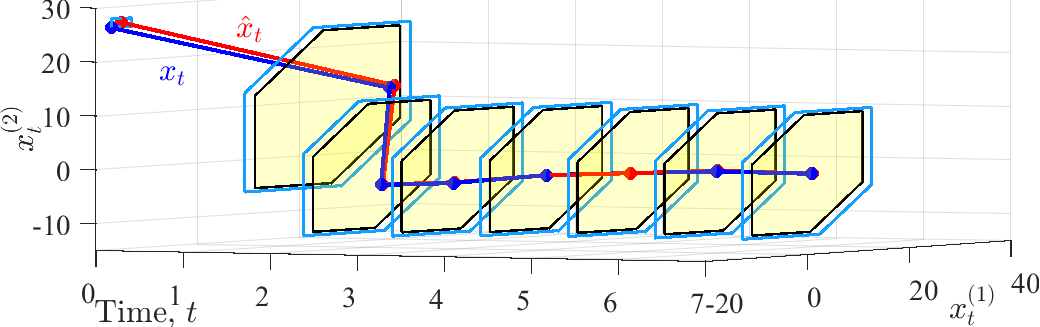}}}
\caption{Inner tube (yellow) and outer tube (blue outline) containing the state estimate (red) and state (deep blue) trajectories, respectively.} 
\label{sisoset1}
\end{figure}


   
Compared to \cite{anch}, where a single set $\widetilde{\mathbb X}\ni\tilde{x}_t$ $\forall t\in\mathbb{I}_0^\infty$ is used to obtain $\widehat{\mathbb X}$, and the terminal set is a subset of $\widehat{\mathbb X}$, here, we get a larger terminal set since $\widehat{\mathbb{X}}_{\text{TS}}\subseteq\widehat{\mathbb{X}}_N\supset \widehat{\mathbb X}$. An increase in the size of $\widehat{\mathbb{X}}_{\text{TS}}$ results in an increased initial feasible region leading to reduced conservatism. This is demonstrated using a SISO system setup having 2 states (for ease of visualization) with $||x_t||_\infty\leq 40$ and $||u_t||_\infty\leq 10$. For the sake of a fair comparison with \cite{anch}, no external disturbances have been considered, i.e., $\mathbb D=\{0_2\}$. The set $\Psi$ has 3 vertices $\psi^{[1]}=[-4.11\;\; 2.2\;\; 15\;\; -10]^\intercal$, $\psi^{[2]}=[-4.12 \;\;2.2 \;\;15.01$ $ -10]^\intercal$ and $\psi^{[3]}=[-4.13\;\; 2.2\;\; 15.02\;\; -10]^\intercal$, and the feedback gain $K=[0.2548 \;\;  -0.0426]$. The initial feasible regions shown in Fig. \ref{sisoset} are computed for $N=7$ following \cite[Sec.~10.3]{borrelli2017predictive}. For a given setup, the initial feasible region of the proposed method is larger in size than that for \cite{anch}. {\color{black}{In addition, we show the two-tube structure for the same setup using the proposed method in Fig. \ref{sisoset1} with $\psi=0.9\psi^{[1]}+0.1\psi^{[2]}$, $F=[0.03 \;\;1;\;\;0.2\;\;0]$, $\hat \psi_0=0.1\psi^{[1]}+0.9\psi^{[3]}$, $x_0=[-17\;\;26.14]^\intercal$, $\hat x_0=[-15\;\;27]^\intercal$ and $||d||_\infty\leq 0.2$.}}

\section{Conclusion}
This work extends the theory of AOFMPC in \cite{anch} to MIMO systems subject to bounded disturbances. The proposed method uses only input and output measurements to design an adaptive observer that simultaneously estimates the plant states and parameters online, while being robust to disturbances. The estimates obtained are used in a reformulated COCP which determines the control inputs and a homothetic tube for the state estimate trajectory. By adding the sets that contain the state estimation errors to the corresponding homothetic tube sections, we obtain an outer tube containing the true state trajectory. This framework ensures both constraint satisfaction and robust exponential stability. Future efforts will focus on leveraging the adaptive design to enhance the feasibility region and reduce conservatism.

\section*{\hypertarget{App1}{Appendix I}}
The proof of Lemma \ref{L01} follows from Theorem 4.11.4 in \cite{ioannou2006adaptive}, except for a few minor technicalities. A sketch is provided below for completeness.

Using \eqref{gd1}, the parameter estimation error is rewritten as
\begin{align}\label{Ateo}
    \tilde{p}_t=\begin{cases}\tilde{p}_{t-1}-\kappa  \phi_t e_t,\;\;\;&\text{if $\bar{p}_{t}\in\Pi$}\\
    \tilde{p}_{t-1}-\kappa  \phi_t e_t+(\bar{p}_t-\bar{\xi}_t),\;\;&\text{otherwise}
    \end{cases},
\end{align}
where $ \bar{\xi}_t\triangleq \text{arg}\min_{\xi\in \Pi} ||\bar{p}_t-\xi||$. Choosing $V_t=\tilde{p}_t^{\intercal}\tilde{p}_t/2$ and using \eqref{Ateo}, we get
\begin{align}
    V_t=\;&\frac{1}{2}( \tilde{p}_{t-1}-\kappa  \phi_te_t +g_t)^{\intercal}( \tilde{p}_{t-1}-\kappa \phi_te_t +g_t)\;\;\;\;\;\;\;\; \nonumber\\
    =\;&\frac{1}{2}\tilde{p}_{t-1}^{\intercal}\tilde{p}_{t-1}+\frac{1}{2}\kappa^2e_t^{\intercal}\phi_t^{\intercal}\phi_te_t+\frac{1}{2}g_t^{\intercal}g_t \nonumber\\
    &\;\;\;-\kappa\tilde{p}_{t-1}^{\intercal}\phi_t  e_t +g_t^{\intercal}(p-\hat{p}_{t-1}  -\kappa  \phi_t e_t) \nonumber \\
    =\;&\frac{1}{2}\tilde{p}_{t-1}^{\intercal}\tilde{p}_{t-1}+\frac{1}{2}\kappa^2e_t^{\intercal}\phi_t^{\intercal}\phi_te_t+\frac{1}{2}g_t^{\intercal} g_t \nonumber\\
    &\;\;\;-\kappa\tilde{p}_{t-1}^{\intercal} \phi_t e_t +g_t^{\intercal}(p-\bar{p}_{t}), \;\;\text{[from \eqref{Ateo}]}, \nonumber\end{align} where \begin{align}
     g_t=\begin{cases}
        0, \;&\text{if }\bar{p}_t\in\Pi\\
        \bar{p}_t-\bar{\xi}_t,\;&\text{otherwise}\label{Ateth}
        \end{cases}.
\end{align}
The difference $\Delta V_t\triangleq  V_t-V_{t-1}=\kappa^2e_t^{\intercal}\phi_t^{\intercal}\phi_te_t/2+ g_t^{\intercal} g_t/2 -\kappa\tilde{p}_{t-1}^{\intercal} \phi_t e_t+g_t^{\intercal}(p-\bar{p}_{t})$. 

From \eqref{Ateth}, the terms containing $g_t$ in $\Delta V_t$ become $0$ when $\bar{p}_t\in\Pi$. For cases when $\bar{p}_t\notin\Pi$, we have 
\begin{align*}
&\;\frac{g_t^{\intercal} g_t}{2}+g_t^{\intercal}(p-\bar{p}_{t})\\
=&\;\frac{g_t^\intercal g_t}{2}-g_t^{\intercal} (\bar{p}_{t}-p)\\
=&\;\frac{g_t^\intercal g_t}{2}-g_t^{\intercal} (\bar{p}_{t} - \bar{\xi}_t+\bar{\xi}_t-p)\\
    =&\;\frac{g_t^\intercal g_t}{2}-g_t^\intercal g_t-g_t^{\intercal} (\bar{\xi}_t-p)\\
=&\;-\frac{g_t^\intercal g_t}{2}-||\bar{p}_t-\bar{\xi}_t|| ||\bar{\xi}_t-p||\cos{\theta}    \leq\;0 ,
\end{align*}{\color{black}{
where $\theta$ is the angle between the vectors $\bar{p}_t-\bar{\xi}_t$ and $\bar{\xi}_t-p$.}} 

Therefore, irrespective of the value of $g_t$, we have 
\begin{align*}
\Delta V_t\leq &\; \frac{\kappa^2e_t^{\intercal}\phi_t^{\intercal}\phi_te_t}{2}-\kappa e_t^{\intercal}\phi_t^{\intercal} \tilde{p}_{t-1}\\
=&\;\frac{\kappa^2e_t^{\intercal}\phi_t^{\intercal}\phi_te_t}{2} -\kappa e_t^{\intercal}(\Gamma_t^2e_t-CF^t\tilde{x}_0-\eta_t)\\
= &\;-\kappa e_t^{\intercal} \left(I_q- \frac{\kappa{\phi_t^{\intercal}\phi_t\Gamma^{-2}_t}}{2}\right)\Gamma_t^2e_t  +\kappa e_t^{\intercal} (CF^t\tilde{x}_0+ \eta_t)\\
\leq &\; -\kappa e_t^{\intercal}\kappa_0\Gamma_t^2e_t +\kappa ||e_t||||\Gamma_t||\Delta_0 \;(\because \phi_t^{\intercal}\phi_t\Gamma_t^{-2}\prec I_q)\\
\leq &\; -\frac{\kappa_0 e_t^{\intercal}\Gamma_t^2e_t}{2}+\frac{\kappa\Delta_0^2}{2\kappa_0},
\end{align*}
where $\kappa_0\triangleq(1-\kappa/2)$, and $\Delta_0$ is the upper bound of $||\Gamma_t^{-1}||||\eta_t+CF^t\tilde{x}_0||$ which is finite by definition of $\Gamma_t$, ${\Delta_1}$ and $\Delta_2$. The remaining analysis easily follows from Theorem 4.11.4 in \cite{ioannou2006adaptive}.

\section*{References}
\bibliographystyle{IEEEtran}
\footnotesize{\bibliography{IEEEabrv,reference}}

\end{document}